\documentclass[preprint,report]{elsarticle}

\usepackage{array,xspace,multirow,hhline,tikz,colortbl,tabularx,booktabs,fixltx2e,amsmath,amssymb,amsfonts,amsthm}
\usepackage{paralist}
\usepackage[boxed]{algorithm}
\usepackage[noend]{algorithmic}
\usepackage{amssymb,amsthm,xspace}

\usepackage{verbatim,ifthen}
\usepackage{enumitem}
\usepackage{url}
\usepackage{pifont}
\usepackage{ifthen}
\usepackage{calrsfs,mathrsfs}
\usepackage{bbding,pifont}
\usepackage{pgflibraryshapes}
\usetikzlibrary{arrows}
\usepackage{centernot}
\usetikzlibrary{decorations.pathreplacing}
\usetikzlibrary{patterns}
\usepackage{pgflibraryshapes}
\usetikzlibrary{positioning,chains,fit,shapes,calc}
	\usepackage{varioref}
	\usepackage{nicefrac}
	\usepackage{hyperref}
\usepackage{cleveref}
\usepackage{tikz}
\usetikzlibrary{topaths,calc}

	\usepackage{tabularx}
	\usepackage{mathtools}

	\algsetup{linenodelimiter=\,}
	\algsetup{linenosize=\tiny}
	\algsetup{indent=2em}

\definecolor{light-gray}{gray}{0.9}

\newtheorem{definition}{Definition}%

\usepackage{boxedminipage}
\usepackage{xspace}

				\newcommand{\bt}[1][]{\ensuremath{\ifthenelse{\equal{#1}{}}{\mathit{BT}}{\mathit{BT}(#1)}}\xspace}

			\newcommand{\pref}{\ensuremath{\succsim}}
			\newcommand{\spref}{\ensuremath{\succ}}

	\newtheorem{lemma}{Lemma}%
	\newtheorem{theorem}{Theorem}%
	\newtheorem{proposition}{Proposition}%
	\newtheorem{corollary}{Corollary}%
	\newtheorem{example}{Example}
	\newtheorem{assumption}{Assumption}
	\newtheorem{remark}{Remark}
	
\usepackage{eqparbox}

	\usepackage{enumitem}
	\setenumerate[1]{label=\rm(\it{\roman{*}}\rm),ref=({\it\roman{*}}),leftmargin=*}
	\newlength{\wordlength}

\usepackage{enumitem}
\setenumerate[1]{label=\rm(\it{\roman{*}}\rm),ref=({\it\roman{*}}),leftmargin=*}

\newcommand{\nbh}[1][]{
	\ifthenelse{\equal{#1}{}}{\nu}{\nu(#1)}
}

\newcommand{\cstr}[1][]{
	\ifthenelse{\equal{#1}{}}{\mathscr S}{\cstr(#1)}
}

\newcommand{\choice}[1][]{
	\ifthenelse{\equal{#1}{}}{\mathit{C}}{\choice(#1)}
}

\begin{document}

		  \title{How long is a piece of string? \\An exploration of multi-winner approval voting and ballot-length restrictions}


	\author{Barton E. Lee} \ead{barton.e.lee@gmail.com}
	
	\address{Data61, CSIRO and the School of Economics, UNSW,\\ Sydney, Australia}
	
%

\begin{abstract}

Multi-winner approval elections are seen in a variety of settings ranging from academic societies and associations to public elections. In such elections, it is often the case that ballot-length restrictions are enforced; that is, where voters have a limit on the number of candidates which they can vote for. Despite this common feature, there does not seem to be any theoretical justification for ballot-length restrictions (Laslier and Van der Straeten~\cite{Las16}).

This work endogenously derives the set of voter best-response ballot lengths under complete information and with general assumptions on voter utilities and voting rules. These results provide justification for some ballot-length restrictions observed in practice, however when considering equilibrium outcomes our analysis shows that this justification is no longer valid. Equilibrium analysis is considered for voters with lazy and truth-bias second-order tendencies and the equilibrium solution concept is pure-Nash equilibria. 

The key insights show that ballot-length restrictions or institutional features which make voting costly may lead to instability in election outcomes when voters have diverse preferences, via the non-existence of equilibria. On the other hand, when equilibria do exist they satisfy desirable properties which are not guaranteed by equilibria attained under costless voting and in the absence of ballot-length restrictions. In summary our results highlight a stark trade-off between stable and desirable election outcomes.

\end{abstract}

\maketitle




\section{Introduction}






In this paper we focus on multi-winner approval elections whereby voters submit approval ballots over the candidate set and based on these a subset of candidates with some predetermined size is elected i.e. the `winning committee'. Approval ballots are unranked ballots which voters use to express a subset of candidates which they `approve' of. The size, or length, of a given voter's approval ballot refers to the number of candidates which they approve of and this number may vary across voters based on their preferences. 

When voters have dichotomous preferences there is no issue with limiting voter preference elicitation to approval ballots (Brams and Fishburn~\cite{BsjFnj78}) however in many settings, this is not a realistic assumption. Without assuming dichotomous preferences, the restriction of preference elicitation to approval ballots induces strategic behaviour which can lead to a voter varying both the set of candidates which she approves and the length of her approval ballot (Niemi~\cite{Nrg84}).

Multi-winner approval voting is currently used in academic societies, public elections and the election of international officials. Academic societies such as the American Mathematical Society, the Institute of Management Science and the Institute of Electrical and Electronics Engineers (IEEE) utilise this voting method. The latter has over 400,000 members as of 2016 making it the world's largest technical professional organisation. Approval-based methods are also used in the election of seven regional government officials in Zurich, and for the election of the secretary-general of the United Nations.

A common feature observed in practice is the enforcement of ballot-length restrictions; that is, an upper bound on the number of candidates which a voter may include in her approval ballot. For example the approval-based election for the regional government representatives in Zurich restricts the number of candidates a voter can approve to seven (Lachat et al. \cite{LrLjfSkvd15}) - which is also the number of candidates to be elected. Another extremely common example is plurality voting whereby a `one-vote-per-person' rule is applied and a single candidate is elected. However, it has been noted in the literature that there ``does not seem to [be] any specific theoretical property'' to justify such restrictions (Laslier and Van der Straeten~\cite{Las16})\footnote{Under a restricted setting of voter preferences which are also dichotomous Aziz et al.~\cite{Aziz17} show some desirable properties of not enforcing ballot-length restrictions. Also Elkind et al.~\cite{Elk17} show that the Bloc approval-based voting rule which requires voters to approve of precisely the number of candidates to be elected satisfies a desirable axiom called fixed majority. However, both \cite{Aziz17} and \cite{Elk17} consider a non-strategic environment.}.

This work considers strategic voters in a complete information, multi-winner approval election with general and heterogenous utilities. Our results provide a characterisation of voter best-response approval ballots lengths - leading to a potential justification of ballot-length restrictions equal to the number of candidates to be elected. In particular, this provides some justification for the restrictions enforced by plurality voting and the previously mentioned elections in Zurich. This result applies to a class of approval voting rules which includes the most standard `AV-rule' whereby the candidates receiving the highest number of approvals are elected as winners until the winning committee (of predetermined size) is filled. 

However our results, when considering (pure-Nash) equilibrium analysis, show that this justification is misguided. If voters are incentivised to submit `shorter' approval-ballots (i.e. approving a minimal number of candidates which guarantees maximal utility) via increases to the cost of voting or the enforcement of ballot-restrictions, then equilibria do not exist unless there is a high degree of consensus among voter preferences. This suggests some level of instability. Conversely when no such incentives exists; for example when voting is costless and in the absence of ballot-length restrictions, then we show that an equilibrium always exists. The effect of these incentives on voter behaviour is formalised via two commonly applied models of voters with second-order tendencies referred to as laziness and truth-bias -- see for example Desmedt and Elkind~\cite{DyEe10}, Dutta and Laslier~\cite{DbLjf10}, Elkind et al.~\cite{Elk15}, Endriss~\cite{End13}, and Xia and Conitzer~\cite{XiCv10}. The former generalises the idea that if a voter can not affect the election outcome they will abstain from voting, the latter implies that a voter which multiple best-responses will strictly prefer a sincere best-response\footnote{It is shown in section 5.2 that at least one best-response ballot is also sincere.}. These refinements lead to lazy-Nash and sincere-Nash (pure) equilibria.

Interestingly, we show that the relatively few equilibria which do exist, when voting is costly and/or ballot-length restrictions are applied, satisfy desirable properties which are not guaranteed to be satisfied when voting is costless and in the absence of ballot-length restrictions. Thus, our results highlight a trade-off between achieving stability and satisfying other desirable features of election outcomes. 


\paragraph{Contributions} 
Our contributions are three-fold: firstly, we present a general yet analytically tractable model to study voter behaviour in multi-winner elections utilising a well-known class of voter utility functions; secondly, we provide insights into the effect of ballot-length restrictions via analysis of best-response behaviour and equilibria analysis; and thirdly, our results extend and complement the existing literature on single-winner approval elections.

\paragraph{Outline and structure}
Section 2 provides a brief survey of the related literature, section 3 introduces the model formally and then section 4 provides our key analysis of voter best response behaviour in the presence of ballot-length restrictions. Lastly, section 5 considers equilibrium analysis of the standard `AV-rule' when voters are either lazy or truth-biased. This provides a heuristic comparison of the effect of ballot-length restrictions and costly voting on equilibrium outcomes.


\section{Related literature}



Strategic voting in single-winner elections with preferential voting has been heavily studied by previous scholars. However strategic voting in two closely related areas; multi-winner elections and, approval-based elections have received relatively little attention individually, let alone when studied in combination. Multi-winner and single-winner elections are considered in Myserson~\cite{Myer93}, and Cox~\cite{Cox85, Cox87} for approval voting however voters are assumed to be sincere\footnote{Strictly speaking the voters do follow a strategy, however it is determined solely by their own cardinal utilities over candidates and do not depend on beliefs, or knowledge, of other voters' strategies.}. Brams and Fishburn~\cite{BsjFnj78} consider strategic voters in single-winner approval elections where voters have dichotomous preferences and both Niemi~\cite{Nrg84} and De Sinopoli et al.~\cite{DeSinop06} provide analysis for voters with non-dichotomous preferences. Three distinct incomplete information models of single-winner approval-based elections with strategic voters have been formulated and studied by Myerson and Weber~\cite{MrbWrj93}, Laslier~\cite{Las09} and Myerson~\cite{Myer1, Myer2, Myer3}. However none of these papers, shed light on whether their results extend to the multi-winner setting, nor the effect of restricting voter strategy spaces, via ballot-length restrictions. One notable exception which considers strategic voters in multi-winner approval-based elections is Laslier and Van der Straeten~\cite{Las16} which will be discussed in greater detail at the end of this section\footnote{A number of computational and algorithmic papers consider multi-winner approval elections, however their focus differs substantial from this paper; some examples include Aziz et al.~\cite{AGGMMW15}, Endriss~\cite{End07, End13}, and  Procaccia and Rosenschein~\cite{Proc07}.}.



Elkind et al.~\cite{Elk15} study plurality (single-winner) elections with strategic voters, focusing on the effect of three different tie-breaking rules on voter behaviour and pure Nash equilibria (PNE). The work analyses two models of strategic voters (lazy and truth-bias models) with complete information which builds on the work of Desmedt and Elkind~\cite{DyEe10} and Obraztsova et al.~\cite{Obr13}. Our work does not focus on the effect of different tie-breaking rules but nonetheless complements this body of work by providing some more general results for multi-winner approval elections which can be viewed as a generalisation of plurality voting. In our equilibria analysis (section 5), we focus exclusively on lexicographic tie-breaking and present a series of results for the multi-winner setting. When restricting these results to the special case of single-winner voting we attain a theorem of Elkind et al.~\cite{Elk15} as a corollary (see \Cref{ifandonlyifk=1} in section 5).

Laslier and Van der Straeten~\cite{Las16} consider strategic voting in multi-winner elections under a specific voting rule and the assumption that voter utilities are given by an additive sum (or equivalently an average utility). The paper focuses on an incomplete information score uncertainty model which is related to the work of Myerson and Weber~\cite{MrbWrj93}. Our work focuses on the complete information setting but allows for a general class of voting rules, does not restrict voters to have additive sum utilities and considers voters with second-order tendencies.  In addition our results provide insights into the relationship between voter utilities and the number of candidates which they will optimally vote for, and the effect of ballot-length restrictions.


\section{Model}



\subsection{Voters, candidates and preferences}


This subsection provides a formal introduction to the model, voter utility functions and characterisations of these functions. Some preliminary lemmas are also stated which highlight the implications of our assumptions upon voter utilities over committees (election outcomes) - these will be utilised in later sections.

Let $N$ be a set of voters and $C$ be a set of candidates with representative elements $i$ and $c$, respectively. For each voter $i\in N$ we express cardinal utilities over the candidate set $C$ via a utility function
\begin{align*}
u_i: C&\rightarrow \mathbb{R}\\
c&\mapsto u_i(c).
\end{align*} 
We shall assume voters have strict preferences; that is, for all $c, c'\in C$ such that $c\neq c'$ we have $u_i(c)\neq u_i(c')$. Sometimes it will be convenient to express voter preferences by the (strict) ordinal relation $\spref_i$ such that for $c, c'\in C$
$$c\spref_i c' \iff u_i(c)> u_i(c').$$
For clarity we have assumed that all voters have strict preferences however, all of the results can be naturally extended to cover weak preferences $\pref_i$ by considering equivalence classes of candidates for each voter $i\in N$.

Under approval-based voting each voter $i\in N$ submits an \emph{approval ballot} $A_i\subseteq C$ which is a subset of candidates that she will declare her `approval' for. Note that $A_i$ is an unranked ballot and the size, or length, of the ballot $|A_i|$ may vary across voters. We refer to the list $A=(A_1, \ldots, A_n)$ of approval ballots as the \emph{ballot profile}. 

We will consider approval-based multi-winner voting rules that take as an input a tuple $(N, C, A, k)$ of voters $N$, candidates $C$, a ballot profile $A$ and a positive integer $k\le |C|$. The output or election outcome is a subset $W\subseteq C$ of size $k$. We will refer to the set $W$ as the \emph{winning set} or \emph{winning committee}. Throughout the committee size will be denoted by $k$. Definitions and results will often be presented without reference to the parameter $k$ but should be understood that $k$ is some fixed positive integer no greater than $|C|$.

In the multi-winner approval setting a committee, say $W\subseteq C$ with $|W|=k\ge 1$, is elected and hence to describe voter preferences over election outcomes, rather than candidates, it is necessary to extend each voter's utility function $u_i(\cdot)$ to a set-extension utility $U_i(\cdot)$. A \emph{set-extension} utility provides a functional form which extends a utility function from candidates to the real line, 
$$u_i:C\rightarrow \mathbb{R},$$
 to a utility function from subsets of candidates of size $k$ to the real line,  
$$U_i:2^C|_{k}\rightarrow \mathbb{R},$$
where $2^C|_{k}$ denotes the subset of the power-set of $C$ which contains sets of size $k$.  

In this paper we assume OWA (order weighted average) set-extensions which were introduced by Yager~\cite{YaR88}. OWA utilities are of the following form: For a given voter $i\in N$, let $W\subseteq C$ of size $k$ and consider a relabelling of the elements of $W$ such that 
$$W=\{c_1, \ldots, c_k\},\qquad \text{where $c_j\spref_i c_{j+1}$ for all $j<k$.}$$
Voter $i$'s OWA set-extension utility, derived from $u_i:C\rightarrow \mathbb{R}$, is
\begin{align}\label{genutil}
U_i(W)=\sum_{j=1}^k \lambda_j^{(i)} u_i(c_j) \qquad \text{where $\lambda_j^{(i)}\ge 0 \quad \text{for all } j\le k$}.
\end{align}
Note that each coefficient $\lambda_j^{(i)}$ is not associated with any particular candidate but rather with the position of a candidate within the committee $W$(with respect to $\spref_i$).

To avoid degenerate cases we assume at least one $\lambda_j^{(i)}> 0$. This assumption is simply for convenience, if for some voter $i\in N$ $\lambda_j^i=0$ for all $j$ our analysis would not be affected since any such voter would be indifferent between all election outcomes and could be assumed to optimally not participate in the election i.e $A_i=\emptyset$. Note that the utility can be compactly written as a dot product in $\mathbb{R}^k$
\begin{align}\label{genutil2}
U_i(W)=\lambda^{(i)}\cdot \hat{u}_i(W),
\end{align}
where $\lambda^{(i)}=(\lambda_1^{(i)}, \ldots, \lambda_k^{(i)})$ and $\hat{u}_i(W)=(u_i(c_1), \ldots, u_i(c_k))$ is an ordered (descending) vector of utility values attained from each candidate in $W$. Thus, when 
\begin{enumerate}
\item $\lambda_1^{(i)}=1$ and $\lambda_j^{(i)}=0$ for all $j> 1$ we attain the \emph{best set-extension}; whereby a voter only attains utility from their most preferred candidate in $W$,
\item  if $\lambda_k^{(i)}=1$ and  $\lambda_j^{(i)}=0$ for all $j< k$ we attain the \emph{worst set-extension}; whereby a voter only attains utility from their least preferred candidate in $W$, 
\item if $\lambda_j^{(i)}=1$ for all $j$ then we attain the \emph{natural additive} utility (or equivalently the average utility); whereby a voter's utility from $W$ is simply the sum of the utility of each candidate in $W$.
\end{enumerate}
More generally we only enforce the nonnegative restriction that $\lambda_j^{(i)}\ge 0$ for all $j\le k$ as in (\ref{genutil}). This restriction has the immediate implication that a voter always weakly prefers a committee $W$ over $W'$ (i.e. $U_i(W)\ge U_i(W')$) if $W$ is attained from swapping a candidate $c'\in W'$ with a more preferred candidate $c\in W$. This is formalised in the following lemma.

\begin{lemma}\label{weakmonoutil}
Let $W \subseteq C$ and suppose there exists $c\in W, c'\notin W$ such that $c'\spref_i c$ then
$$U_i(W \cup\{c'\} \backslash \{c\})\ge U_i(W).$$
\end{lemma}

\begin{proof}
Consider the dot product interpretation of voter utilities (\ref{genutil2}). First note that for any fixed vector $\lambda^{(i)}\in \mathbb{R}^k$ such that $\lambda^{(i)}\ge \boldsymbol{0}$ (component-wise inequality) the dot product function 
\begin{align*}
f: \mathbb{R}^k&\rightarrow \mathbb{R}\\
\boldsymbol{x}&\mapsto \lambda^{(i)}\cdot \boldsymbol{x}
\end{align*}
is a weakly increasing function. Further, if  $c'\spref_i c$ then $\hat{u}_i(W \cup\{c'\} \backslash \{c\})\ge \hat{u}_i(W)$ (component-wise inequality). Combining these two facts completes the proof.
\end{proof}

Notice that despite our assumption that voters have strict preferences over $C$, a voter's OWA set-extension utility $U_i$ over committees need not be strict.

\begin{definition}\label{j-set-util}\emph{[$j^*$ set-extension]}\\
Given a voter $i$ with OWA utility set-extension, if $j^*$ is the smallest integer such that
$$\lambda_\ell^{(i)}=0 \qquad \text{for all } \ell> j^*,$$
then we say that voter $i$ has a $j^*$ set-extension utility. If no such $j^*$ exist we define $j^*=k$ by default.
\end{definition}

Applying this definition we see that the best-set extension is a $j^*=1$ set-extension, whilst both the worst-set extension and the natural additive utilities are $j^*=k$ set-extensions.

 In the definition above we have denoted voter $i$'s set-extension value by $j^*$. As will be seen later, this paper allows distinct voters in $N$ to have distinct set-extension values; that is, heterogenous utility functions and set-extension values. When this is important will we denote a voter $i$'s set-extension value by $j^*(i)$.




\subsection{Multi-winner voting rules}


In this paper we consider a class of multi-winner approval voting rules which includes the commonly used AV-rule whereby the $k$ candidates with the highest approval scores are elected, and lexicographic tie-breaking is used. The class is a subfamily of non-degenerate best-$k$ scoring rules. Best-$k$ scoring rules were introduced and characterised by Elkind et al.~\cite{Elk17}. Our focus on this class derives from two monotonicity properties -- relative rank monotonicity and monotonic robustness (to be defined later) -- satisfied by such (non-degenerate) rules which will later be shown to lead to a characterisation of voter best-response behaviour.

We begin by providing a formal definition of best-$k$ scoring rules under approval-based voting. For further details, properties and a more general definition of best-$k$ scoring rules we refer the reader to \cite{Elk17}.

\begin{definition}\label{def1}\emph{[Best-$k$ score rule (approval-based)]}\\
Given a candidate set $C$, a set of voters $N$ and a ballot profile $A$, let $f(c, A)$ be a (scoring) function mapping a candidate and ballot profile pair $(c,A)$ to a real number such that a weak order over $C$ can be formed. A multi-winner voting rule which elects the candidates in the top $k$ positions into the winning committee $W$ is a best-$k$ score rule (deterministic tie-breaking when necessary).
\end{definition}

The standard AV-rule is a member of the best-$k$ scoring rule family. To fit with the above definition we define the scoring function as
\begin{align}\label{av}
f(c,A)=|\{i\in N\, :\, c\in A_i\}|.
\end{align}
Thus, $f(c,A)$ denotes the number of voters who included the candidate $c$ in their approval ballot $A_i$. This is referred to as an approval score. The best-$k$ rule then select the $k$ candidates with highest $f(c,A)$ values (i.e. highest approval scores) using a tie-breaking rule where necessary. However, the family of best-$k$ scoring rules also includes dictatorial rules whereby a single voter determines the entire election, or predetermined election rules whereby the candidates to be elected are determined before (or independently) of the approval ballot profile $A$. These two rules are in some sense degenerate and for this reason we focus on what we call non-degenerate best-$k$ scoring rules. 

To avoid degenerate voting rules we assume that all voters are treated equally (but candidates need not be). First, given an approval ballot profile $A$ define
$$S(c, A)=\{i\in N: c\in A_i\}\qquad \text{ and } \qquad s(c, A)=|S(c,A)|.$$
That is, $S(c,A)$ denotes the set of voters approving of candidate $c$ under the ballot profile $A$ and $s(c,A)$ denotes the approval score (or equivalently, the number of approving voters). We define the following subfamily of non-degenerate voting rules.

\begin{definition}\emph{[Non-degenerate best-$k$ score rule (approval-based)]}\\
Let $f(c, A)$ be a (scoring) function for a best-$k$ score rule. The function is said to be non-degenerate if in addition to \Cref{def1}
\begin{enumerate}
\item the function only depends on the candidate, and the number of approvals
$$f(c,A)=g(c, s(c,A))$$
\item the function $g: C\times \mathbb{R} \rightarrow \mathbb{R}$ is strictly increasing in the second argument and $g(c, 0)=0$ for all $c\in C$. 
\end{enumerate}
\end{definition}

Note that any non-degenerate best-$k$ scoring rule necessarily ensures anonymity, hence a dictatorial rule is degenerate. Also the condition $g(c, 0)=0$ implies that a predetermined election rule is also degenerate. The standard AV-rule (\ref{av}) is of course non-degenerate. Other less-standard election rules are also non-degenerate for example; candidates need not be treated equally this would be achieved with the following scoring function
$$f(c, A)=\lambda(c)\, s(c,A),$$
where $\lambda:C \rightarrow \mathbb{R}_{>0}$.

We now define the key monotonicity properties ensured by non-degenerate best-$k$ rules which will be crucial for characterising voter behaviour.


Recall that a voter $i\in N$ submits an approval ballot denoted by $A_i\subseteq C$. Suppose $c\in C$ such that $c\notin A_i$ (i.e. candidate $c$ is not approved by voter $i$ in $A_i$) we call the alternate approval ballot
$$A_i'=A_i\cup\{c\},$$
a \emph{reinforcement} of $c$ with respect to $A_i$. We say that a ballot profile $A'=(A_1', \ldots, A_n')$ \emph{reinforces} $c$ relative to another ballot profile $A=(A_1, \ldots, A_n)$ if there exists some ballot $A_i'$ which reinforces $c$ relative to $A_i$ and all other ballots are unchanged. Throughout we shall use the convention that given a ballot profile $A$ we denote the election outcome $W=W_A$.




We now define two monotonicity properties which are satisfied by all non-degenerate best-$k$ voting rules and are sufficient conditions for our results to apply.
 
The first monotonicity property captures the standard notion of candidate monotonicity by requiring that a candidate $c\in W_A$ must still be elected under any alternate ballot $A'$ which reinforces $c$. However, the property also extends to instances where another candidate, say $c'\in C$, receives less approvals and then ensures that $c$ must still be elected. Informally the property captures the notion that if an elected candidate's approval score weakly increases, say $c\in W_A$, whilst all other candidates receive weakly lower approvals under $A'$ then it must still be the case that $c\in W_{A'}$. 

\begin{definition}\label{mon3}\emph{[Relative rank monotonicity]}\\
Let $(N, C, A, k)$ be a voting instance with outcome $W_A$ and let $c\in W_A$. Let $A'$ be a ballot profile such that 
$$s(c, A')\ge s(c, A)\qquad \text{ and } \qquad  s(c', A')\le s(c',A)\qquad \forall c'\in C\backslash \{c\},$$
then $c\in W_{A'}$.
\end{definition}



The second property we introduce is called monotonic robustness\footnote{Monotonic-robustness implies $1$-robustness in the sense of Bredereck et al.~\cite{Bred17}.} This property captures the idea that if only a single candidate $c\in C$ position is improved then the only possible change to the winning committee is that $c$ becomes elected. If $c$ was already elected then monotonic robustness implies candidate monotonicity.

\begin{definition}\label{mon2}\emph{[Monotonically-robust]}\\
Let $(N, C, A, k)$ be a voting instance with outcome $W_A$, let $c \in C$ and $A'$ be a reinforcement of $c$. A voting rule is said to be \emph{monotonically robust} if whenever
$$W_A\neq W_{A'} \implies W_{A'}=\{c\}\cup \Big(W_A\backslash \{c'\}\Big) \qquad \text{for some $c'\in W_A$}.$$
\end{definition}




\begin{proposition}
All non-degenerate best-$k$ voting rules satisfy relative rank monotonicity and monotonic robustness.
\end{proposition}

\begin{proof}
Given a non-degenerate best-$k$ voting rule with (scoring) function $f$ we shall show that both properties are satisfied. 

We begin with relative rank monotonicity: Consider voting instance with outcome $W_A$ and $c\in W_A$ and let $A'$ be a ballot profile such that
$$s(c,A')\ge s(c,A) \quad \text{ and }\quad s(c', A')\le s(c', A)\quad \forall c'\in C\backslash \{c\}.$$
Since $c\in W_A$ it follows that $f(c, A)$ is among the top-$k$ scores (when taking the tie-breaking rule into account). Now consider $f(c,A')=g(c, s(c,A'))$, since $g$ is strictly increasing in the second argument it follows that 
$$f(c, A')\ge f(c,A).$$
Furthermore, for any $c'\in C\backslash \{c\}$ it follows that $f(c', A')\le f(c', A)$.Thus, if $c$ was among the top-$k$ scores under $A$ it must also be among the top-$k$ scores under $A'$ -- we conclude that $c\in W_{A'}$. 

We now prove that monotonic robustness holds: Let $c\in C$ and let $A'$ be a reinforcement of $c$ relative to ballot profile $A$ -- this means that $s(c, A')=s(c, A)+1$ and $s(c', A')=s(c', A)$ for all other $c'\in C$. Thus, $f(c', A)=f(c',A')$ for all $c'\neq c$ and $f(c, A')> f(c, A)$. It follows that if $W_A\neq W_{A'}$ it must be that $f(c, A)$ was not among the top-$k$ scores under $A$, but $f(c, A')$ is among the top-$k$ score sunder $A'$. Furthermore, this is the only change and so the single candidate $c'\in C$ with the $k$-th best score under $A$ is replaced in $W_{A'}$ with the reinforced candidate $c$. This completes the proof.
\end{proof}

Some of our work applies to voting rules which only satisfy rank relative monotonicity but not necessarily monotonic robustness -- any such voting is necessarily outside of the family of non-degenerate best-$k$ rules. Our proofs are always proven with explicit reference to which properties are required, thus it will be clear when monotonic robustness is not required.


\section{Ballot-length restrictions and minimal best response approval ballots}


This section considers the impact of ballot-length restrictions on voter behaviour - in particular, whether and when such a restriction will prevent a voter from submitting their otherwise `optimal' ballot.

Ballot-length restrictions are common in practice however, it is unclear for what purpose (Laslier and Van der Straeten~\cite{Las16}). A potential argument is that ballot-length restrictions simply reduce the number of votes without affecting the election outcome. Indeed, as seen in the previous section a voter may be indifferent between all ballot entries and so a ballot-length restriction may assist in reducing the complications around vote counting whilst not affecting a voters best-response action and the election outcome. 

We show that ballot-length restrictions can in some instances prevent a voter from submitting any ballot which would be a best-response in the absence of such restrictions. In fact whether, or not, a voter is affected directly by the restriction can be determined via the value of their $j^*$ set-extension. 

Interestingly we find that when the ballot-length restriction equals the number of candidates to be elected no voter will ever be prevented from submitting a ballot which is also a BR ballot in the absence of restrictions.  This coincides with ballot-length restrictions observed in real-world elections such as plurality voting or the regional Zurich election (mentioned in the introduction) whereby seven winning candidates are to be elected and voters can approve of at most seven candidates.

We begin by formally defining a best response (BR) and minimal best response (MBR) approval ballot and a ballot-length restriction.

\begin{definition}\emph{[Best response (BR) and minimal best response (MBR) ballots]}\\
Given a voter $i\in N$, an approval ballot $A_i$ is a best-response (BR) ballot to $A_{-i}$ if 
$$U_i(W_{(A_i, A_{-i})})\ge U_i(W_{(A_i', A_{-i})})\qquad \forall A_i'\in 2^C,$$
where $2^C$ is the power set of $C$. We denote the set of best-response ballots to $A_{-i}$ for voter $i$ by $B_i(A_{-i})$. A ballot $A_i\in B_i(A_{-i})$ is minimal best response (MBR) ballot if 
$$|A_i|\le |A_i'| \quad \forall A_i'\in B_i(A_{-i}).$$
\end{definition}


\begin{remark}
The size of a minimal best-response ballot is unique. However, minimal best-response ballots need not be unique even under the strict preference assumption over $C$. With the additional assumption of `full-rank' set-extensions MBR ballots become unique (this additional assumption is introduced and discussed in section 5). 
\end{remark}

We now define an $R$ ballot-length restriction which is a feature of voting rule which limits the size of valid voter approval ballot i.e. $|A_i|\le R$. We also define an $R$ ballot-length restriction to be constraining if there exists instances where a voter is prevented from submitting an approval ballot $A_i$ which would be a best response if no restrictions was enforced; that is $|A_i|>R$.

\begin{definition}\label{constraining}\emph{[$R$ ballot-length restriction]}\\
A voting rule which only permits approval ballots $A_i$ such that $|A_i|\le R$ for some positive integer $R$, is referred to as an $R$ ballot-length restriction. If $R\ge |C|=m$ we say the voting rule is unrestricted.

Given a voter $i\in N$, we say that an $R$ ballot-length restriction is \emph{constraining} for voter $i$ if there exists an instance $A_{-i}$ such that for every BR ballot (in the unrestricted setting) $A_i$ we have 
$$|A_i|>R,$$
and hence $A_i$ can not be submitted.
\end{definition}

Note that in the definition above for a constraining $R$ ballot-length restriction we could equivalently write the condition in terms of the size of a MBR ballot $A_i$ since this implies the inequality for all BR ballots.

The following key lemma provides a characterisation of the size of a voters minimal best response ballot based on their $j^*$ set-extension utility in the presence of a ballot-length restriction. The lemma states that a voter who derives utility from some function of their top $j^*$ most preferred candidates in a committee can achieve maximal utility, under unilateral deviations, by submitting a ballot of size no larger than $j^*$ in the absence of a ballot-length restriction. However, if an $R$ ballot length restriction is enforced such that $R<j^*$ then there are instances where the voter will be strictly worse off.

 The intuition for the result is that if a voter's best response was to submit a ballot of size greater than $j^*$ but only gains utility from $j^*$ candidates in the committee, say $C'$, then any approval/vote for a candidate not in $C'$ could be removed without reducing the voters utility.

\begin{lemma}\label{lem2}\emph{[The `how long is a piece of string?' Lemma]}\\
Let $i\in N$ be a voter with $j^*$ set-extension. In the absence of a ballot-length restriction voter $i$'s MBR $A_i$ is such that
$$|A_i|\le j^*.$$
If an $R$ ballot-length restriction with $R<j^*$ is enforced then the restriction is constraining -- in particular, there exists voting instances such that voter $i$ will be strictly worse off under unilateral deviations. 
\end{lemma}

\begin{proof}
Let $A_i$ be a BR ballot to $A_{-i}$ with outcome $W$. Label the elements of $W$ according to voter $i$'s (strict) preferences i.e.
$$W=\{c_1, \ldots, c_k\}$$
such that $c_j \spref_i c_{j+1}$ for all $j<k$.

From the $j^*$ set-extension utility we have
\begin{align*}
U_i(W)&=\sum_{j=1}^{k} \lambda_j^{(i)}u_i(c_j)=\sum_{j=1}^{j^*} \lambda_j^{(i)}u_i(c_j),
\end{align*}
and so voter $i$ attains utility from only candidates $\{c_1, \ldots, c_{j*}\} \subseteq W$. Notice that since $W$ is the outcome under a BR ballot $U_i(W)$ is the maximum achievable utility for voter $i$ when facing ballot profile $A_{-i}$.

To prove the first claim it suffices to show that a BR ballot exists of size at most $j^*$. For the purpose of a contradiction suppose that $A_i$ is a MBR ballot of size greater than $j^*$, we shall construct an alternate approval ballot $A_i'$ which is of size $j^*$. Define $A_i'$ as follows
$$A_i'=A_i \backslash (A_i\backslash \{c_1, \ldots, c_{j^*}\}),$$
note that any candidate in $A_i\backslash \{c_1, \ldots, c_{j^*}\}$ is removed from $A_i$ to form $A_i'$. Let $W'$ be the corresponding outcome from $A'=(A_i', A_{-i})$. In this election, the approval scores of candidates in $A_i\backslash \{c_1, \ldots, c_{j^*}\}$ strictly decrease, whilst all other candidates' approval scores are unchanged. Thus, by rank-relative monotonicity it must be the case that $\{c_1, \ldots, c_{j^*}\}\subseteq W'$. It follows that 
\begin{align}\label{pieceofstringlemeq1}
U_i(W')\ge U_i(W)
\end{align}
since voter $i$'s top $j^*$ most preferred candidates from $W$ i.e. $\{c_1, \ldots, c_{j^*}\}$, are still contained in $W'$. However, since $W$ provides maximal utility it must be that equality holds in (\ref{pieceofstringlemeq1}), thus $A_i'$ is also a BR ballot. Furthermore,
$$|A_i'|\le j^*.$$
This is a contradiction since $A_i$ was assumed to be a MBR ballot, yet $A_i'$ achieves the same utility but is of strictly smaller size. 

For the second claim, consider an $R$ ballot-length restriction with $R<j^*$. Let $C'\subseteq C$ be voter $i$'s top $j^*$ preferred candidates in $C$. Suppose that voter $i$ faces a ballot profile $A_{-i}$ such that for all $c\in C$ 
$$f(c, A_{-i})=0.$$
That is all voters in $N\backslash \{i\}$ submit an empty approval ballot and so voter $i$ can completely determine the election outcome. If the voting rule uses a tie-breaking method which leads to $W\cap C' =\emptyset$ then voter $i$'s MBR ballot is 
$$A_i=C',$$
of size $j^*$. However, this is not a valid ballot since an $R$ ballot-length restriction ($R<j^*$) is enforced. It follows that voter $i$'s BR among valid ballots leads to strictly less utility than if the ballot-length restriction was not enforced.
\end{proof}


We now present a series of corollaries which follow immediately from \Cref{lem2} and \Cref{constraining}. These corollaries have direct policy implications regarding the assumptions required on voter utilities to justify a ballot-length restriction as being non-constraining. With the exception of \Cref{ss} the assumptions are strong and unlikely to hold in reality.

The first corollary states that if voters can vote for no more than the number of candidates to be elected (i.e $R=k$) then the ballot-length restriction is non-constraining. As mentioned previously, this coincides with real-world approval elections observed in practice. The intuition is as follows: if committee of size $k$ is to be elected then each voter's utility depends on some function of these $k$ candidates. If a voter had a best response ballot of size greater than $k$, under complete information, at least one of these approvals/votes is unnecessary since only $k$ candidates are being elected. 

\begin{corollary}\label{ss}
If $R=k$ then the ballot-length restriction is non-constraining. In particular, plurality voting (single-winner) is non-constraining.
\end{corollary}

The second corollary states that if the common restriction of `one-vote-per-person' is applied to a multi-winner election then this is only non-constraining to voters in one situation whereby voters only derive utility from their most preferred candidate elected into the committee -- this is indeed a strong assumption.

\begin{corollary}
Plurality voting in a multi-winner setting is non-constraining if and only if all voters have $j^*=1$ set-extension utility. 
\end{corollary}

The third corollary states that a ballot-length restriction is only non-constraining for voters who derive utility from a relatively small number of candidates in the committee (i.e. low $j^*(i)$ set-extension value). This suggests that voters with relatively narrow utilities with respect to the winning committee; for example only attaining utility from their top few most preferred candidates in $W$, are less likely to be affected by a ballot-length restrictions. It follows that ballot-length restrictions will disadvantage voters who derive utility holistically from the winning committee.

\begin{corollary}
Consider an $R$ ballot-length restriction, this is non-constraining when every voter has a $j^*(i)$ set-extension utilities such that $j^*(i)\le R$ for all $i\in N$. If this is not the case, it will be constraining for the subset of voters with $j^*(i)> R$.
\end{corollary}

We conclude that ballot-length restrictions of size $R<k$ will affect different voters to different extents, and because of this inequality, it is unclear whether ballot-length restrictions are ever fully justified without strong assumptions upon voter $j^*(i)$ set-extension values.


\section{Equilibrium analysis: costly voting and ballot-length restrictions}


This section presents analysis for two behavioural tendencies of voters; laziness and truth-bias. Lazy voting refers to when voters would rather abstain then to vote when their vote makes no difference to the election outcome, whilst truth-bias voting has voters who prefer to choose sincere actions (to be defined) among their set of strategically optimal actions (so long as this is still a best-response). Both of these tendencies have been presented previously in the literature\footnote{See for example Desmedt and Elkind~\cite{DyEe10}, Dutta and Laslier~\cite{DbLjf10}, Elkind et al.~\cite{Elk15}, Endriss~\cite{End13}, and Xia and Conitzer~\cite{XiCv10}.} however, we extend them to the multi-winner setting and provide equilibrium analysis in this new setting. In the multi-winner setting lazy voting is characterised by voters submitting minimal best-response ballots - since this is the minimal amount of `effort', or votes, required to attain maximal utility. 

The purpose of our focus on lazy and truth-bias voting is to capture the impact of institutional features which increase the cost of voting or enforce ballot-length restrictions. Lazy voting tendencies can be modelled via positive but sufficiently small cost to voting as in Xia and Conitzer~\cite{XiCv10}, in the extreme case a ballot-length restriction may force a voter to submit a minimal best-response ballot by making any other best-response ballot invalid. In the less extreme case, it may simply encourage lazy voting. The opposite tendency is truth-bias which is more likely to occur when voting is costless and ballot-length restrictions are absent. If ballot-length restrictions are present a sincere and best-response ballot may not be valid.

Our key insight shows that when voters have diverse preferences a lazy equilibrium is unlikely to exist whilst under truth-bias voting an equilibrium always exists. This suggest a sense of instability when voters are lazy which can be inferred as an indirect and adverse effect of increases to the cost of voting or the introduction of ballot-length restrictions. However, and perhaps surprisingly, when lazy equilibria do exists they are guaranteed to satisfy a desirable property which is not guaranteed under a truth-bias equilibrium. 





As mentioned we consider voters with behavioural tendencies, or secondary preferences. This modelling choice is in part motivated by the observation of several scholars that (Nash) equilibrium analysis of strategic voting often shows a plethora of equilibria - many which seem `unreasonable' and are not useful as predictors of behaviour. One method for refining these equilibria is to add a secondary preference, or tendency, of voters such that when there are multiple optimal actions at the voter's disposal they select an action based on this secondary preference. This additional feature also addresses the key difficulty highlighted by Cox~\cite{Cox87} when analysing equilibria of approval voting systems as summarised in the following quote ``voters have such a wide choice of ballots, it is not clear, in general how to forecast their votes simply on the basis of their preference ranking of candidates". 

The following example illustrates what seem to be `unreasonable' pure-Nash equilibria (PNE).

\begin{example}
Let $C=\{a, b, c\}$, $N=\{1, 2, 3\}$, $k=1$ and suppose all voters have the following strict preferences \begin{align*}
\spref_1, \spref_2, \spref_3:& \qquad b,\, a,\, c.
\end{align*}
That is, all voters prefer candidate $b$ to $a$ and prefer candidate $a$ to $c$. If the candidate with the highest number of approvals is elected (i.e the standard `AV-rule') then a PNE of $W=\{c\}$ can be supported by all voters submitting the approval ballot $A_i=\{c\}$. This is a PNE despite all voters having identical preferences and $c$ being the least preferred candidate. Similar approvals can be constructed to show that any other candidate in $C$ can also be elected in equilibrium.
\end{example}

All results that follow\footnote{One exception is \Cref{prop2} which holds more generally for non-degenerate best-$k$ rules and is stated with this more general condition.} are for the standard `AV-rule', whereby the $k$ candidates with highest approval scores are elected and ties are broken via a lexicographic priority ordering (predetermined and deterministic). We shall denote the lexicographic ordering/relation by $\triangleright$. This specific rule belongs to the family of non-degenerate best-$k$ scoring rules and hence the results from previous sections hold. 

\begin{assumption}
The results in this section focus on the standard `AV-rule', whereby the $k$ candidates with highest approval scores are elected and ties and broken via a lexicographic priority ordering. This rule belongs to the family of non-degenerate best-$k$ scoring rules. This assumption will be assumed without being explicitly stated in the remaining results.
\end{assumption}


\subsection{Lazy voting and lazy Nash} 


In this subsection we focus on `lazy' voting whereby a voter wishes to abstain from voting when her vote can never make a profitable change to the election outcome. By extension, when it is profitable to vote for multiple candidates lazy voting requires that voters choose a utility maximising ballot of the smallest cardinality, or length, as possible. That is, a lazy voter always chooses a minimal best response (MBR) ballot. 

We begin by formally defining lazy voting and lazy equilibria. We then turn to analysis of the existence and properties of lazy-equilibria in the multi-winner setting. We provide a dichotomous characterisation of lazy-equilibria in relation to a set of `most preferred candidates' which depends of voter $j^(i)$ set-extension values. This highlights a desirable property which is guaranteed by lazy-equilibria. This result is followed by a necessary and sufficient condition for lazy-equilibria in one of the characterisation and a necessary condition in the other characterisation. When translating these results to the single-winner setting we attain a result of Elkind et al.~\cite{Elk15} as a corollary.


\begin{definition}\emph{[Lazy voting]}\\
We shall say voters engage in \emph{lazy voting} when given $A_{-i}$ they choose an approval ballot $A_i$ which is a minimal best response (MBR) ballot to $A_{-i}$. That is, given two best-response ballots $A_i$ and $A_i'$, $A_i$ is strictly preferred to $A_i$ if 
$$|A_i|<|A_i'|.$$
\end{definition}

Note that the lazy voting action only applies to best-response ballots and hence can be viewed as a second-order tendency since the preference for shorter approval ballots is only applied after first maximising utility from the election outcome.

\begin{definition}\emph{[Lazy-PNE]}\\
An election outcome $W$ is a lazy-PNE supported by ballot profile $A$ if no voter has an incentive to unilaterally deviate under lazy voter preferences. We denote an equilibrium pair by $(W, A)$.
\end{definition}

It follows that a lazy-PNE $(W,A)$ occurs if and only if for every voter $i\in N$, $A_i$ is a MBR ballot to $A_{-i}$.


Recall that the number of approvals a candidate $c\in C$ attains under ballot profile $A$ is denoted by $s(c, A)$. Also recall that under the standard `AV-rule' in a $k$ multi-winner election the candidates with $k$ highest approval scores $s(c,A)$ are elected (applying lexicographic tie-breaking when necessary).

The following proposition provides a useful observation of lazy-equilibrium approval scores. The observation is that at most $k$ candidates receive precisely one approval and all other candidates receive zero in equilibrium. The result is intuitive since voters are lazy and if $k$ candidates are to be elected -- any vote for an unelected candidate can not be a minimal best response. Furthermore, if any of the $k$ elected candidates, say $c$, has more than one approval votes -- since unelected candidates have zero votes -- a lazy voter with a vote for $c$ could remove this and not change the election outcome under unilateral deviations.

\begin{proposition}\label{lazyprem}
Under a lazy-PNE, say $(W, A)$,
$$s(c, A)=0\quad \forall c\notin W.$$
 That is, at most $k$ candidates can receive a positive share of approval votes. Furthermore,
 $$s(c, A)\le 1\quad \forall c\in W.$$
\end{proposition}

\begin{proof}
For the first statement: 
Assume $(W,A)$ is a lazy-PNE and suppose for the purpose of a contradiction that $s(c', A)>0$ for some $c'\notin W$. Let $i\in N$ be a voter such that $c'\in A_i$, note that $A_i$ must be a MBR ballot. Now suppose that $i$ submits the alternate approval ballot $A_i'=A_i\backslash \{c'\}$, then based on the relative rank monotonicity property the winning committee does not change. Thus, it must be that $A_i$ was not a MBR which is a contradiction.

For the final statement: 
Applying the first statement we have that $s(c', A)=0$ for all $c'\notin W$, it follows that at most $k$ candidates receive a positive share of approvals. Now for the purpose of a contradiction suppose that $s(c, A)\ge 2$ for some $c\in W$. Let $i\in N$ be a voter such that $c\in A_i$, note that $A_i$ must be a MBR ballot. Now if voter $i$ submits the alternate ballot $A_i'=A_i\backslash \{c\}$, then only candidate $c$ receives a change in approval score. Furthermore, candidate $c$ still has a positive number of approvals (i.e. at least $1$ approval) and so is still among the top $k$ candidates in terms of highest approvals $s(\cdot, A)$. It follows that candidate $c$ is still elected and the election outcome is unchanged. Thus, it must be that $A_i$ was not a MBR which is a contradiction. 
\end{proof}


It was noted in Desmedt and Elkind~\cite{DyEe10} that a lazy-PNE need not exist under the single-winner plurality voting setting. This setting is a special case of our multi-winner approval election when setting $k=1$. Notice that any ballot-length restriction $R$ has no effect under the assumption of lazy voting in the single-winner setting ($k=1$). We provide a statement of the $k=1$ result below and simply note that it is easy to construct examples where the same conclusion is reached for $k>1$.

This non-existence of lazy-equilibria, or instability of lazy voting, stems from voter incentives to abstain (or vote for a minimal number of candidates) from the election when their vote is not pivotable. This in turn means that other voters are more likely to be pivotable since in effect there are less votes and less `participating' voters. This push and pull of incentives for lazy voters to vote minimally but also exploit instances where their vote is pivotable can lead to examples where a lazy-equilibrium does not exist.


\begin{proposition}[See Desmedt and Elkind~\cite{DyEe10}]
A lazy-PNE need not exist even in single-winner elections (i.e. $k=1$).
\end{proposition}


We now define a voter's `ideal' set which will later be used to characterise lazy-equilibria. 

Given a voter $i\in N$ with $j^*(i)$ set-extension utility we define voter $i$'s ideal set as
\begin{align}\label{iideal}
W_i^*=\{\text{voter $i$'s top $j^*(i)$ most preferred candidates in $C$}\}.
\end{align}
It is easy to see that if voter $i$'s ideal set is contained in the election outcome, $W_i^*\subseteq W$, then voter $i$ attains maximal utility from the election outcome $W$. 

The following lemma shows that in fact a strict increase in utility is attained when moving from an election outcome $W'$ such that $W_i^*\not\subseteq W'$ to an election outcome $W$ such that $W_i^*\subseteq W$. The proof is left to the appendix however, the intuition is simple; a voter $i$ who derives utility from their top $j^*(i)$ preferred candidates in the committee $W\subsetneq C$ can do no better than having their top $j^*(i)$ preferred candidates in $C$, i.e. $W_i^*$, elected.


\begin{lemma}\label{genutilstrict}
Let $i\in N$ be a voter with $j^*(i)$ set-extension utility and ideal committee $W_i^*$. If $W, W'\subseteq C$ of size $k$ and $W_i^*\subseteq W$ but $W_i^*\not\subseteq W'$, then 
$$U_i(W)>U_i(W').$$
\end{lemma}

As seen in \Cref{genutilstrict}, a voter $i$ with her ideal set $W_i^*$ contained in the election outcome strictly prefers this to any other outcome not containing $W_i^*$. It is natural to suspect that when $W_i^*$ is not contained in $W$ but contains more elements of the ideal set $W_i^*$ than $W'$, then $U_i(W)>U_i(W')$. This however is not true in general. For example consider a voter $i$ with preferences
$$\spref_i: \qquad a, b, c, d$$
$k=2$, and a $j^*(i)=2$ set-extension utility such that $\lambda_1^{(i)}=0$ then her utility is
$$U_i(W)=\lambda_1^{(i)}u_i(c_1)+\lambda_2^{(i)}u_i(c_2)=0\,u_i(c_1)+u_i(c_2).$$
 That is, voter $i$ attains utility only from their least preferred candidate in the election outcome of size $k=2$. Thus $W_i^*=\{a, b\}$ however
$$U_i(\{c, d\})=U_i(\{b, d\})=U_i(\{a, d\})=u_i(d).$$
Thus, we do not attain strict preference for committees which contain more ideal set members $W_i^*$. Note, however that this is no longer the case if $\lambda_1^{(i)}>0$. This provides (technical) motivation for considering cases where $\lambda_1^{(i)}>0$ -- or more generally, a voter with $j^*$ set-extension has $\lambda_j^{(i)}>0$ for all $j\le j^*$. We call this requirement a `full-rank' set-extension and define the concept formally below. This is utilised to provide a characterisation of the existence of lazy-PNE.


\begin{definition}\emph{[Full-rank $j^*$ set-extension]}\\
Let voter $i\in N$ have a $j^*(i)$ set extension. If in addition 
$$\lambda_j^{(i)}>0\qquad \forall j\le j^*(i),$$
then we say that voter $i$ has a \emph{full-rank} $j^*(i)$-set extension.
\end{definition}

Notice that the full-rank assumption is automatically satisfied in the single-winner setting (i.e. $k=1$). Thus, the consideration of full-rank and non-full-rank set-extension utilities is a unique challenge faced in multi-winner election. Also note that for any non-full-rank set-extension utility it can be approximated arbitrarily well by a full-rank set-extension. 

The following lemma states a more general version of  \Cref{genutilstrict} for when a voter strictly prefers one committee over another, under the full-rank assumption. The proof is left to the appendix. The intuition of the result is similar to that of  \Cref{genutilstrict} and is implicit in the example provided as motivation for full-rank set-extension utilities.

\begin{lemma}\label{fullrankpref}
Let $W'\subseteq C$ such that $c^*\in W_i^*$ and $c^*\notin W'$ then if
$$W=W'\cup\{c^*\}\backslash\{c\}\qquad \text{for some $c\in W'$ such that $c^*\spref_i c$},$$
we have
\begin{align}\label{frankeq}
U_i(W)\ge U_i(W').
\end{align}
Furthermore, under the full-rank assumption strict inequality holds for (\ref{frankeq}).
\end{lemma}



Given a (strict) preference profile $\spref=(\spref_1, \ldots, \spref_n)$ and profile of set-extension values $(j^*(1), \ldots, j^*(n))$ , we define the \emph{ideal set of all voters}
\begin{align*}
W^*&=\cup_{i\in N} W_i^*\\
&=\cup_{i\in N} \{\text{voter $i$'s top $j^*(i)$ candidates in $C$}\}.
\end{align*}
This is the union of all voters most preferred committee $W_i^*$ (introduced in (\ref{iideal})). 


The size of the set $W^*$ is a measure of the diversity of voter preferences modified by their $j^*(i)$ set-extension value. For example if $|W|^*\le k$ there is a high degree of consensus among voter preferences (after modifying for set-extension values) and this can be accommodated with a $k$ sized election. Whilst if $|W^*|>k$ then voters have diverse preferences and no election outcome of size $k$ can accomodate all voter preferences.

Under the full-rank assumption, we show that the only lazy-PNE are when $W^*\subseteq W$ or $W\subsetneq W^*$. Despite previous results showing undesirable outcomes of lazy voting such as the non-existence of equilibria - this results highlights a close connection between lazy-PNE and the top preferences of voters. This can be viewed as a desirable property of lazy-PNE when they exist. 

The intuition for the result is as follows; in equilibrium no candidate, say $c'$, outside of the ideal set $W^*$ can be elected, for if this were not the case then there would exist some unelected candidate $c\in W^*$ and at least one voter will be able to profitably change their approval ballot to elect $c$ instead of $c'$. The lazy tendency of voters is crucial as it means in equilibrium at most $k$ candidates receive precisely one approval vote and all other candidate receive none (\Cref{lazyprem}) -- and so, this ensures a deviation which elect $c$ instead of $c'$ to exist. The full-rank assumption is then required to ensure such a deviation is strictly profitable for the deviating voter (see \Cref{fullrankpref}). Without the full-rank assumption the voter may be indifferent between such a deviation, even if the voter $i\in N$ is such that $c\in W_i^*$.

\begin{lemma}\label{keylemma}
Under the full-rank assumption, if $W$ is a lazy-PNE then either
\begin{align}\label{lemmaeq1}
W^*\subseteq W \qquad \text{ or }\qquad W\subsetneq W^*.
\end{align}
Note that if the winning committee $W$ is to be of size $k$ then if $|W^*|\le k$ only the former is possible, and if $|W^*|>k$ only the latter is possible.
\end{lemma}

\begin{proof}
Suppose for the purpose of a contradiction that $W$ is a lazy-PNE supported by the ballot profile $A$ but (\ref{lemmaeq1}) does not hold. That is, there exists $c, c' \in C$ such that
\begin{align}
&c\in W^*\qquad c\notin W\\
 &c'\notin W^*\qquad c'\in W.
\end{align}
Now since the tie-breaking priority $\triangleright$ is a complete ordering either $c\triangleright c'$ or $c'\triangleright c$.

In both instances if $W$ is a lazy-PNE it must be that $s(c', A)=0$. To see this note that since $c'\in W$ we have $s(c', A)\le 1$ (\Cref{lazyprem}), if $s(c',A)=1$ then there exists a voter $j$ such that $c'\in A_j$ and $c'\notin W_j^*$. Furthermore there are at most $k$ candidates with scores equal to $1$. Now either $W_j^*\subseteq W$ or not. In the first case, voter $j$ has their most ideal set elected which does not include $c'$ and so removing this vote does not affect voter $j$'s utility and so under lazy tendencies voter $j$ will not submit $c'\in A_j$. Alternatively, suppose that there exists $c^*\in W_j^*$ but $c^*\notin W$, then clearly $s(c^*,A)=0$ (\Cref{lazyprem}). But now, if voter $j$ submits the alternate approval ballot
$$A_j'=W\cup\{c^*\}\backslash \{c'\}$$
then there will be precisely $k$ candidates $W\cup\{c^*\}\backslash \{c'\}$ with positive approval scores $s(\cdot, A')$. It follows that the new election outcome will be 
$$W'=W\cup\{c^*\}\backslash \{c'\},$$
which is strictly preferred to $W$ by voter $i$ (\Cref{fullrankpref}). Thus we have a contradiction that $A_i$ is a MBR ballot.

We conclude that $s(c', A)=0$, at most $(k-1)$ candidate receive non-zero approval scores and $s(c, A)=0$ (since $c\notin W$).

In the first instance, this gives a contradiction since both candidates $c$ and $c'$ have zero approvals but $c'$ is elected despite having a lower lexicographic priority rank.

In the second instance, a voter say $i\in N$ with $c\in W_i^*$ could submit the alternate ballot
$$A_i'=W\cup\{c\}\backslash\{c'\}.$$
This will lead to precisely $k$ candidates i.e. $W\cup\{c\}\backslash\{c'\}$ receiving positive approval scores $s(\cdot, A')$, It follows that the new election outcome will be
$$W'=W\cup\{c\}\backslash\{c'\},$$
which is strictly preferred to $W$ by voter $i$ (\Cref{fullrankpref} and applying the full-rank assumption) since $c'\notin W_i^*$. Thus we have a contradiction that $A_i$ is a MBR ballot. 
\end{proof}

We now consider the first case of a lazy-equilibrium where $W^*\subseteq W$. The following theorem shows that if $|W^*|\le k$, there is a unique lazy-PNE $W$ which contains $W^*$. The intuition is straightforward; if every voter has their most ideal set of candidates in $C$ elected into $W$ (hence $W^*\subseteq W$) then every voter's utility from $W$ is independent of candidates $c'\notin W^*$ (even if $c'\in W$). Thus, since voters are lazy every such candidate $c'\notin W^*$ will receive zero votes. It follows that any candidate $c'\notin W^*$ but $c'\in W$ is elected based on the (lexicographic) tie-breaking rule applied to the set of candidates with zero approval scores -- this necessarily leads to a unique set of candidates outside of $W^*$ being elected.

Before formally presenting the result, we first introduce some notation which will be utilised in a number of results of this subsection.

Label the candidate set $C$ according to their lexicographic ordering ($\triangleright$) such that
\begin{align}\label{doublehash}
c_1\triangleright c_2\triangleright\cdots \triangleright c_m,
\end{align}
where $|C|=m$. Similarly given a committee $W\subseteq C$ of size $k$, label the elements according to the lexicographic ordering, say $W=\{w_1, \ldots, w_k\}$. However, it will be convenient to introduce a mapping $\sigma: [k]\rightarrow [m]$ which translates the elements of $W$, say $w_j$, to an element, say $c_{\sigma(j)}\in C$.  That is, we write the elements of $W$ as
\begin{align}\label{labelling}
W&=\{w_1, \ldots, w_k\}\notag\\
&=\{ c_{\sigma(1)}, \ldots, c_{\sigma(k)}\}.
\end{align}
such that $\sigma(\ell)<\sigma(\ell')$ for all $\ell<\ell'$ and so $c_{\sigma(\ell)}\triangleright c_{\sigma(\ell')}$.


\begin{theorem}\label{containthm}
The committee $W\subseteq C$ of size $k$ such that
$$W^*\subseteq W,$$
is a lazy-PNE if and only if for every $c\in W-W^*$ has a higher priority than every $c'\notin W$.
\end{theorem}

\begin{proof}
Label the elements of $W^*$ according to a mapping $\sigma: [|W^*|]\rightarrow [m]$ (as was done in  (\ref{labelling})), i.e. 
$$W^*=\{c_{\sigma(1)}, \ldots, c_{\sigma(\ell)}\},$$
where $\ell\le k$.

Now we construct a ballot profile $A$ which satisfies the necessary conditions for $W$ such that $W^*\subseteq W$ to be lazy-PNE. We start with $c_{\sigma(\ell)}$. If $\sigma(\ell)\le k$ then all candidates in $W^*$ have a top $k$ lexicographic ranking $\triangleright$ in $C$ i.e. $W^*\subseteq \{c_1, \ldots, c_k\}$, thus $s(c, A)=0$ for all $c\in C$ is a lazy-PNE. This holds since all voters have their ideal set $W_i^*$ elected simply by the priority $\triangleright$ when no votes are cast, and so no voter will have an incentive to deviate. Also all voters are submitting MBR ballot since $A_i=\emptyset$ for all $i\in N$. In addition since $W_i^*\subseteq W$ for all $i\in N$, there can be no equilibrium whereby a candidate $c\in W$ and $c\notin W^*$ has $s(c, A)>0$. Thus, $s(c, A)=0$ for all $c\notin W^*$ and the candidates in $W-W^*$ will be elected if and only if the second condition holds.

Otherwise, $\sigma(\ell)>k$ and some candidates in $W$ are outside of the top $k$ lexicographic rankings in $C$. To construct a ballot profile $A$ which supports $W$ as a lazy-PNE we do the following: initialise $A_i=\emptyset $ for all $i\in N$,  select a voter $i$ such that $c_{\sigma(\ell)}\in W_i^*$ and set $A_i\mapsto A_i\cup\{c_\ell\}$. Now consider candidate $c_{\sigma(\ell-1)}$. If $\sigma(\ell-1)\le k-1$ then we are done since $\{c_{\sigma(1)}, \ldots, c_{\sigma(\ell-1)}\}\subseteq \{c_1, \ldots, c_{k-1}\}$ and $W^*$ will be elected by the tie-breaking priority $\triangleright$ if no additional votes are cast. Otherwise, $\sigma(\ell-1)>k-1$ and we select a voter, say $i\in N$, such that $c_{\sigma(\ell-1)}\in W_i^*$ and set $A_i\mapsto A_i\cup\{c_{\ell-1}\}$. Repeat this procedure for all candidates in $W^*$.

This is a lazy-PNE since all voters have their most preferred set $W_i^*$ elected - removing any approval will change the outcome to be strictly less desirable (recall \Cref{genutilstrict}) and hence there is no incentive to elect any alternate candidates (clearly no combination of these is optimal either). Again the candidates in $W-W^*$ are also elected if and only if the second condition holds.
\end{proof}

We now present another key theorem which considers the second case where $W\subsetneq W^*$, or equivalently $|W^*|>k$. The result gives necessary conditions for an election outcome to be a lazy-PNE when the winning committee does not contain the ideal set $W^*$; that is, $W\subsetneq W^*$. This result when restricted to the $k=1$ case, will then lead immediately to a theorem presented in Elkind et al.~\cite{Elk15}. 

The result is dependent on the lexicographic tie-breaking rule applied. The first part of the theorem shows that the candidate elected into $W$ in a lazy-equilibrium with the lowest lexicographic tie-breaking preference (or rank) must have a sufficiently low rank. This result may seem counter-intuitive since it implies that if $|W^*|>k$ (i.e. a low degree of consensus among voter preferences) the top-$k$ candidates with respect to the tie-breaking can never form an equilibrium  -- this is stated formally in \Cref{corrrrr}. The intuition for the result is as follows; if voters are lazy and support a high ranked (with respect to the tie-breaking rule) candidate they have an incentive to remove their approval vote and leave the election result up to the tie-breaking rule which their high ranked candidate will win. Under unilateral deviations this is of course a (minimal) best response but this generates an opportunity for other voters to be pivotable an elect a different candidate which they prefer. 

Informally speaking, the second part of the theorem shows that in equilibrium every high ranked (with respect the tie-breaking rule) and unelected candidate, say $\tilde{c}$, must be unanimously less preferred to the lowest ranked and elected candidate $c'$\footnote{Technically speaking we only need unanimity among voter who derive utility from candidate $c'$. That is, $c'$ is among voter $i$'s top $j^*(i)$ most preferred candidates in $W$.}. The intuition follows from the fact that at most $k$ candidates in equilibrium receive non-zero approval scores and the rest receive zero approvals in equilibrium (\Cref{lazyprem}). Thus, any voter under a unilateral deviation can elect the high ranked $\tilde{c}$ candidate -- if $W$ is an equilibrium it must necessarily be the case that all voters agree that such a candidate $\tilde{c}$ is less preferred to the lowest ranked and elected candidate $c'\in W$




We now present the theorem formally but briefly recall some notation, the elements of $W$ can be expressed via the map $\sigma$ as follows
\begin{align}\label{doublehash2}
W&=\{w_1, \ldots, w_k\}\notag\\
&=\{ c_{\sigma(1)}, \ldots, c_{\sigma(k)}\}.
\end{align}
such that $\sigma(\ell)<\sigma(\ell')$ for all $\ell<\ell'$ and so $c_{\sigma(\ell)}\triangleright c_{\sigma(\ell')}$.


\begin{theorem}\label{rlem}
Let $W=\{ c_{\sigma(1)}, \ldots, c_{\sigma(k)}\}\subsetneq W^*$ be lazy-PNE. It must be that $\sigma(k)>k$. Furthermore under the full-rank assumption, for all voters $i\in N$ it must be that if $c_{\sigma(k)}$ is among voter $i$'s top $j^*(i)$ most preferred candidates in $W$ for all $c_j\notin W$ with $j<\sigma(k)$ 
$$c_{\sigma(k)}\spref_i c_j \qquad \forall i\in N.$$
\end{theorem}

\begin{proof}
For the first statement, suppose for the purpose of a contradiction that $W\subsetneq W^*$ is a lazy-PNE supported by the ballot profile $A$ and $\sigma(k)\le k$. This implies that there are $k$ candidates in $W$ with lexicographic priority rank at least $k$ and so it must be that 
$$W=\{c_1, \ldots, c_k\}.$$
From \Cref{lazyprem}, it follows that $s(c, A)\le 1$ for all $c\in W$ and $s(c', A)=0$ for all $c'\notin W$. If $s(c, A)=1$ for any $c\in W$, then it can be shown that the voter $i\in N$ with $c\in A_i$ is not submitting a MBR ballot since the priority ordering ($\triangleright$) of $c$ ensures that the alternate ballot $A_i'=A_i\backslash\{c\}$ produces the same election outcome. We conclude that $A_i=\emptyset$ for all $i\in N$ and hence $s(c, A)=0$ for all $c\in C$. But if this is the case, then every voter $i\in N$ can completely determine the election outcome via a unilateral deviation and since $W$ is a lazy-PNE it follows that $W_i^*\subseteq W$ for all $i\in N$. This contradicts the assumption that $W\subsetneq W^*$.

For the final statement, we begin by noting that if $W$ is a lazy-PNE outcome supported by the ballot profile $A$ then $\sigma(k)>k$. Furthermore, recalling \Cref{lazyprem} it must be that at most $k$ candidates in $W$ receive precisely one approval vote and the remaining candidates receive zero approvals. For the purpose of a contradiction suppose that voter $i\in N$ is such that $c_{\sigma(k)}$ is among her top $j^*(i)$ most preferred candidates in $W$ and there exists $c_j\notin W$ with
$$c_j\spref_i c_{\sigma(k)}.$$
 We will show that voter $i$ can change the election outcome to $W':=W\cup\{c_j\}\backslash \{c_{\sigma(k)}\}$ which is strictly preferred by voter $i$ under the full-rank assumption.

Let $\overline{W}$ and $\underbar{W}$ denote the set of candidates in $W$ who receive an approval score of $1$ and $0$, respectively under $A$. Consider the alternate ballot which voter $i$ could submit 
$$A_i'=\underbar{W}\cup\{c_j\}\cup A_i\backslash \{c_{\sigma(k)}\}.$$
Under the ballot profile $A'=(A_i', A_{-i})$ the approval scores of all $k$ candidates in $W\cup\{c_j\}\backslash \{c_{\sigma(k)}\}$ is precisely one approval and all other candidates, except possibly $c_{\sigma(k)}$, receive zero approvals. 

Candidate $c_{\sigma(k)}$ receives either zero or one approval. In the former case $k$ candidates have non-zero approvals and the rest have zero approvals and so the new winning committee under $A'$ is $W':=W\cup\{c_j\}\backslash \{c_{\sigma(k)}\}$. In the latter case the winning committee is still $W'$ since candidate $c_{\sigma(k)}$ has the lowest tie-breaking ranking among all $k+1$ candidates in $W'\cup\{c_{\sigma(k)}\}$. This shows that voter $i$ can change the election outcome to $W'$. It only remains to show that this deviation is strictly profitable to voter $i$ -- under the full-rank assumption this can be shown to hold in a similar manner to \Cref{fullrankpref} and so is omitted. This contradicts the assumption that $W$ is a lazy-PNE and completes  the proof.
\end{proof}

An immediate corollary of the above theorem is as follows. The intuition was discussed in the paragraph prior to \Cref{rlem}.

\begin{corollary}\label{corrrrr}
$W=\{c_1, \ldots, c_k\}$ is a lazy-PNE if and only if $W^*\subseteq W$. That is, every voter's ideal set is contained in the election outcome.
\end{corollary}

Combining the above statements leads to the following `if and only if' corollary for the single-winner setting ($k=1$) which was shown in Elkind et al.~\cite{Elk15}. Recall the notation from (\ref{doublehash2}). The result shows that for a lazy-equilibrium in single-winner election there must be a high degree of consensus among voters; either all voters prefer a single candidate, or some candidate is elected who is most preferred by at least one voter and candidate with higher tie-breaking rank are unanimously less preferred. 

\begin{corollary}\label{ifandonlyifk=1}
When $k=1$ we have:

$W=\{c_j\}$ is a lazy-PNE if and only if
\begin{enumerate}
\item $c_j$ is the first preferences of all voters
\item $j>1$, $c_j$ is some voters top preference and for all $\ell<j$ we have that
$$c_j\spref_i c_\ell \qquad \forall i\in N.$$
\end{enumerate}
\end{corollary}

\begin{proof}
First note that if $k=1$, then the full-rank assumption holds immediately. 

We begin by proving the forward direction. Suppose $W=\{c_j\}$ is a lazy-PNE, from \Cref{keylemma}, then either $W^*\subseteq W$ or $W\subseteq W^*$. Since $k=1$, for each $i\in N$ the ideal set $|W_i^*|=1$ and so $W^*\subseteq W$ if and only if $W_i^*=c_j$ for all $i\in N$; that is, $c_j$ is every voters first preference. If $W\subsetneq W^*$ then statement (ii) follows from \Cref{rlem}.

We now prove the reverse direction. Suppose (i) holds then $W^*=\{c_j\}\subseteq W$ and so \Cref{containthm} shows that $W$ is a lazy-PNE. Suppose (ii) holds then a lazy-PNE is constructed by letting a voter $i\in N$ with $c_j$ as their top preference submit the ballot $A_i=\{c_j\}$ and all other voters submit $A_{i'}=\emptyset$. 
\end{proof}

One may ask whether there is a relationship, such as set containment, of equilibria when voters have full-rank and non-full-rank utilities. Unfortunately, there does not appear to be an obvious relationship. It is straightforward to construct examples such that an equilibrium with the full-rank assumption need not be an equilibrium without this assumption, and conversely, an equilibrium which holds when the full-rank assumption does not hold need not be an equilibrium when it does. For brevity we omit such an example.

The results of this subsection highlight that under lazy voting, approval voting rules can be inherently unstable if voters have diverse preferences i.e. $|W^*|>k$. As an implication this suggest that institutional features which make voting more costly or encourage voters to submit MBR ballots, say via a ballot-length restriction, may also generate instability. In this setting, the cost of voting is not restricted to financial or time - some recent theoretical and empirical work has looked at the effects of hidden (anonymous voting) or open (non-anonymous voting) polls on voter behaviour (see Obraztsova et al.~\cite{Obr17} and Zou et al.~\cite{Zou15}). The researchers suggest that social factors/benefits can drive voter behaviour and the number of alternatives which they approve. Thus, it is easy to also imagine social factors as being considered as a cost/benefit of voting, especially when elections are highly controversial.


\subsection{Truth-biased voting and sincere Nash} 



As seen in the previous section, `lazy' voting has the undesirable property of instability, via the non-existence of equilibria. In this subsection we consider an alternative tendency of voters to be truthful rather than lazy, called truth-bias. Such a tendency is more likely to arise when voting is costless and a necessary condition is that ballot-length restrictions are not enforced. Thus the analysis in this subsection provides a heuristic comparison of approval-based elections without ballot-length restrictions and when voting is costless, against the results of the previous subsection. The key result shows that an equilibrium always exists under truth-bias tendencies contrasting against the non-existence results shown when voters are lazy. 

We begin by introducing the notion of sincere voting in an approval-based setting. We then show that there always exists a best-response ballot which is also sincere.

Informally, an approval ballot is said to be sincere if no candidate $c$ who is more preferred than an approved candidate $c'\in A_i$ remains unapproved i.e. $c\notin A_i$. This notion was introduced by Brams (1982) (see \cite{Brams82}). A formal definition is provided below.

\begin{definition}\label{sinceredef}\emph{[Sincere voting]}\\
Let $i\in N$ be a voter, we say that an approval ballot $A_i\subseteq C$ is sincere if for all $c\in A_i$ and for all $c'\in C\backslash A_i$ 
$$c\spref_i c'.$$
\end{definition}

We note that in this setting sincere voting is not a single voting action but rather a characterisation of a class of voting actions which are all sincere. In fact, for a given voter with strict preferences $\spref_i$ any sincere ballot can be completely characterised by a threshold number of candidates that are approved. This notion can also be extended to weak preferences by considering ordered equivalence classes of candidates.

We now present a theorem showing the existence of a sincere best-response approval ballot -- this provides additional motivation for considering truth-bias voting as a second-order tendency since it is compatible with utility maximising behaviour under the `AV-rule'\footnote{In fact, this compatibility holds more general for non-degenerate best-$k$ scoring rules.}. The result provides an efficient method for a given voter to construct a sincere and best-response ballot. The proof is left to the appendix, but follows by constructing a sincere and best response ballot -- this is achieved by considering a minimal best response ballot $A_i$ and simply approving of additional candidates such that the approval ballot is now sincere. That is, the voter approves of every candidate in $C$ which is more preferred to a candidate in $A_i$ to construct the sincere and best response ballot $A_i'$. Intuitively this ballot $A_i'$ maintains the voters maximal utility since the voter only approves of more preferred candidates -- the monotonicity properties of the AV-rule ensures that this can not make the voter worse off.


\begin{theorem}\label{prop2}
Let $i\in N$ be a voter and a non-degenerate best-$k$ voting rule. For every ballot profile $A_{-i}$ there exists a sincere and best response (BR) ballot $A_i$ for voter $i$.
\end{theorem}

Notice that for any voter there are just $(m+1)$ sincere approval ballots. Thus, the mere existence of a sincere and BR approval ballot means that a voter can consider just $(m+1)$ approval ballots to determine and find a BR ballot, rather than searching $2^m$ approval ballots. This is an efficient method for voters to construct a BR ballot.

The above theorem shows that submitting a sincere voting is a weakly dominant strategy - however, it does not imply the existence of a sincere equilibrium since sincere voting is not a single action. We also emphasise that the existence of a sincere and best-response ballot does not mean the voting rule is strategy-proof. On the contrary this suggest that even when a voter is restricted to sincere approval ballots, strategic manipulations are still possible (see for example Niemi~\cite{Nrg84}). Thus, the well-known conclusions of the Gibbard-Satterthwaite Theorem~\cite{Gib, Sat}, and the multi-winner analog by Duggan and Schwartz~\cite{DjSt00}, do not apply.

\begin{remark}
As a point of clarification the above result appears to be in conflict with the conclusions of the related work Laslier and Van der Straeten~\cite{Las16}. In particular in Proposition 4 of their paper, it is shown that in general a voter submitting a BR ballot may preclude the possibility of sincere voting. The researchers consider a model of incomplete information (referred to as the score uncertainty model) but show this result as the difference between the incomplete and complete information settings become arbitrarily small. 

The apparent conflict arises due to differences in the definition/description of best response ballots. Under the formulation used by Laslier and Van der Straeten~\cite{Las16} a best response ballot of a voter $i$ is not simply an approval ballot which maximises the expected utility for voter $i$ under unilateral deviation. Rather the best response ballot referred to in \cite{Las16} corresponds to a ballot derived from a sequence of pairwise comparisons and decisions upon which candidate to include in the approval ballot. This formulation of BR is based on a behavioural rule referred to as the `Leader Rule' (see Laslier~\cite{Las09} for further details). Thus, there is no conflict in the result presented in this paper and those presented in Laslier and Van der Straeten~\cite{Las16}.
\end{remark}


Given \Cref{prop2}, it is clear that a voter is never adversely affected by being restricted to, or having a second-order preference for, sincere approval ballots. That is, a sincere and best response approval ballot is guaranteed to exist. Thus, the notion of sincere, or truth-bias, voting is compatible with utility maximising voters. Formally we define truth-bias voting below.

\begin{definition}\emph{[Truth-bias voting]}\\
We shall say voters engage in \emph{truth-bias voting} when given $A_{-i}$ they choose an approval ballot $A_i$ from the set of best response (BR) ballots to $A_{-i}$ which is also sincere.
\end{definition}


The truth-bias action only applies to the set of best-response ballots and hence can be viewed as a second-order tendency since the preference for sincere approval ballots is only applied after first maximising utility from the election outcome. Below we define sincere equilibria which was also considered in \cite{DbLjf10} as an equilibrium refinement.

\begin{definition}\emph{[Sincere-PNE]}\\
An election outcome $W$ is a sincere-PNE supported by ballot profile $A$ if no voter has an incentive to unilaterally deviate under truth-bias voter preferences. We denote an equilibrium pair by $(W, A)$.
\end{definition}



Note that even if a voter $i$ is indifferent between submitting a sincere and an insincere ballot (under unilateral deviations) the choice can have distinct implications for other voters' actions and then in turn produce different election outcomes. 

The next proposition shows that under truth-bias voting an equilibrium is guaranteed to exists, so long as the number of voters exceeds the number of candidates. This result is in contrast to the non-existence results shown when voters are lazy. The intuition for the proposition is as follows: if voters are truth-biased and their ballot does not affect the election result then they have no incentive to deviate (as long as the ballot is sincere). Whilst, under lazy voting in such a situation a voter would remove their ballot entirely which not only generates volatility in voter ballots but also creates greater opportunities for voters to influence the election outcome. 

More technically, the proof of the proposition constructs an equilibrium by constructing ballot profiles which ensures that all voters vote sincerely (possibly by abstaining) and the elected candidates have approval scores of at least two and all other candidates have zero approvals. Thus, no single voter can influence the election outcome and even though every non-abstaining voter need not vote to maintain the election result -- since voters are truth-biased rather than lazy -- there is no incentive to deviate.

\begin{theorem}
Assume $n>m$. A sincere Nash equilibrium always exists.
\end{theorem}

\begin{proof}
Given any positive integer $k$ we provide construct an approval ballot profile which is a sincere-PNE. Let $E\subseteq C$ be the set of candidates who are the first preference of some voter $i\in N$ among candidates in $C$. 

Since $n>m$ by the pigeonhole principle ($|E|\le |C|<m$) there exists at least one candidate, say $c_1^*\in E$, with at least two voters having $c_1^*$ as their first preference among candidates in $C$. Select this candidate $c_1^*$ and for all voters $i$ who have $c_1^*$ as their first preference let $A_i=\{c_1^*\}$. Note that this is indeed a sincere ballot.

Now update the set $E$ to be the set of candidates who are the first preference of some voter $i\in N$ among candidates in $C\backslash \{c_1^*\}$.  The set $E$ includes candidates which are the second preference of voters who had $c_1^*$ as their first preference. 

 The size of $E$ is at most $m-1$ and there are still $n$ voters and so by pigeonhole principle there exists at least one candidate, say $c_2^*$, with at least two voters with $c_2^*$ in their first preference among candidates in $C\backslash \{c_1^*\}$. For all such voters let $A_i\rightarrow A_i\cup\{c_2\}$, again this is a sincere ballot. 
 
Repeating this process a total of $k$ times. We end up with $k$ candidates each with approval scores at least two and only sincere approval ballots. This is a sincere-PNE, since any single voter changing their approval ballot can alter the score of candidates by at most one.

 
If we require that voters submit non-empty approval ballot then for any remaining voters with $A_i=\emptyset$ can simply submit approval ballots $A_i=C$ which does not change any of the relative approval scores and is also sincere.
 \end{proof}

The above proposition shows a sense of stability under truth-biased voting due to the guaranteed existence of equilibria, which was not present under lazy voting. This can be seen as a desirable effect of reducing the cost of voting and not implementing ballot-length restrictions. However, this `stability' comes at a cost since the equilibria outcome may be less reasonable than under lazy voting. As seen in the previous subsection for lazy-PNE there is a close relationship between equilibria $W$ and the ideal set $W^*$ (see \Cref{keylemma}) - this relationship no longer holds under truth-bias voting. The following examples illustrates this point.

\begin{example}\emph{[Lazy vs sincere PNE]}\\
Consider the situation where $k=2$, $C=\{a, b, c, d\}$ and $N=\{1, 2, 3\}$ with full-rank utilities and strict preferences 
\begin{align*}
\spref_1:& \qquad a, \, b, \, c, \, d\\
\spref_2, \spref_3:& \qquad c, \, d, \, a, \, b.
\end{align*}
If $j^*(1)=1$ and $j^*(2)=j^*(3)=2$ then the election $W=\{a, b\}$ is a sincere-PNE supported by 
$$A_1=\{a, b\}\qquad \text{ and }\qquad A_2=A_3=\emptyset.$$
This is despite $W^*=\{a, c, d\}$. However, such an outcome can not be sustained under lazy voting since we require that either $W^*\subseteq W$ or $W\subseteq W^*$. In fact, the unique lazy-PNE is $W=\{a, c\}$.

If voter utilities are considered comparable and voter $2$ and $3$ are identical, we see that the combined utility of voters in the sincere-PNE $W=\{a,b\}$ is 
\begin{align}\label{finaleq1}
\lambda_1^{(1)}u_1(a)+2\Big(\lambda_1^{(1)}u_2(a)+\lambda_2^{(1)}u_2(b)\Big),
\end{align}
whilst under the lazy-PNE $W=\{a,c\}$ the combined utility is 
$$\lambda_1^{(1)}u_1(a)+2\Big(\lambda_1^{(1)}u_2(c)+\lambda_2^{(1)}u_2(a)\Big),$$
which is strictly greater than (\ref{finaleq1}). Thus the lazy-PNE is socially more desirable.
\end{example}

\section{Conclusion}


In this work we considered how voters vary their approval ballot-lengths for differing utility functions. In particular, we characterised the length of minimal best-response (MBR) ballots. MBR ballots are a key object when considering the effects of ballot-length restrictions, since if a MBR ballot exceeds the restriction then no best-response (BR) ballot will be a feasible ballot. Thus, we were able to provide clear results into when a ballot-length restriction would prevent a voter from achieving their maximal utility under unilateral deviations. This led to two key insights; firstly, if ballot-lengths are restricted to the size of the committee to be elected then no voter will ever be prevented from submitting a BR ballot, secondly, voters whose utility does not depend on all of the elected candidates are less likely to be affected by ballot-length restrictions. 

We then presented equilibria analysis for voters with second-order tendencies such as laziness and truth-bias. The key results showed that lazy voting is inherently unstable since equilibria need not exist, whilst the opposite holds for truth-bias voting. This suggest that institutional features which increase the cost of voting or ballot-length restriction may encourage lazy voting and hence generate instability. This however, should be coupled with the fact that equilibria (when they exist) under lazy voting appear to be more desirable.


	\section*{Acknowledgments}
	

	I acknowledge the support from the UNSW Scientia PhD fellowship. I would like to thank Haris Aziz, Juan Carlos Carbajal, Gabriele Gratton, Richard Holden, Hongyi Li,  and Lirong Xia for helpful comments and fruitful discussions. In addition, I thank the participants of the Data61/CSIRO Algorithmic Decision Theory (ADT) group and the UNSW Econ-Theory workshop.

\section*{Appendix}



\begin{proof}[Proof of \Cref{genutilstrict}]
To prove this statement we utilise the dot-product interpretation of utilities (recall Equations (\ref{genutil}) and (\ref{genutil2})) i.e. for a given set $S\subseteq C$ of size $k$
$$U_i(S)=\lambda \cdot  \hat{u}_{S}:=\lambda(i) \cdot \Big(u_i(c_1), \ldots, u_i(c_j), \ldots, u_i(c_{j^*})\Big),$$
where $\lambda(i)\in \mathbb{R}_{\ge 0}^{j^*}$, $\lambda_{j^*}(i)>0$ and $\hat{u}_S$ is an ordered (descending) vector of voter $i$'s utility from their $j^*$ most preferred candidates in $S$. We denote voter $i$'s $j$-th most preferred candidate in $S$ by $c_j\in S$.

Now label the elements of $W$ and $W'$ according $i$'s preferences so that
$$W=\{c_1, \ldots,c_{j^*}, \ldots,  c_k\}\qquad\text{ and }\qquad W'=\{c_1', \ldots, c_{j^*}',\ldots, c_k'\}.$$
Now since $W_i^*\subseteq W$ we have
$$W_i^*=\{c_1, \ldots,c_{j^*}\}.$$
That is, voter $i$'s $j^*$ most preferred candidates in $W$ and $C$ coincide. Now consider the elements of $W'$ for every $j\le j^*$ we have
$$u_i(c_j')\le u_i(c_j).$$
Now suppose that we have equality for all $j$, then it must be that $c_j'=c_j$ for all $j\le j^*$ since preferences are strict. But then $W_i^*\subseteq W'$ which is a contradiction. And so it must be that strict inequality hold for at least one $j\le j^*$, say $j=\ell$ (let $\ell$ be the smallest such positive integer where strict inequality is attained). 

If $\ell=j^*$, then we have $U_i(W)>U_i(W')$ since $\hat{u}_W\ge \hat{u}_{W'}$ (component-wise inequality) and $\hat{u}_W\neq \hat{u}_{W'}$ and the component such that strict inequality is attained occurs at $\ell=j^*$ where $\lambda_{j^*}(i)>0$. 

Now assume $\ell<j^*$, we can attain a similar conclusion since it must be that 
$$u_i(c_{\ell+1}')<u_i(c_{\ell+1}).$$
This follows from the fact that if $u_i(c_j')<u_i(c_j)$ it must be that $u_i(c_j')\le u_i(c_{j+1})$, and so
$$u_i(c_{j+1}')<u_i(c_j')\le u_i(c_{j+1}).$$
Hence the same conclusion follows that $u_i(W)>u_i(W')$.
\end{proof}


\begin{proof}[Proof of \Cref{fullrankpref}]
Label the elements of $W$ and $W'$ according $i$'s preferences so that
$$W=\{c_1, \ldots,c_{j^*}, \ldots,  c_k\}\qquad\text{ and }\qquad W'=\{c_1', \ldots, c_{j^*}',\ldots, c_k'\}.$$
Now $c^*\in W_i^*\cap W$ and suppose $c^*=c_\ell\in W$ for some $\ell\le j^*$. It follows that
$$u_i(c_j)\ge u_i(c_j')\qquad \forall j\le j^*,$$
with strict inequality holding for $j=\ell$ since $c^*=c_\ell \notin W'$. Thus,
$$U_i(W)\ge U_i(W').$$

Now if $\lambda_j>0$ for all $j\le j^*(i)$ and noting that $c^*$ is necessarily among voter $i$'s top $j^*(i)$ most preferred candidates in $W$ we have
$$U_i(W)> U_i(W').$$
Thus, under the full-rank assumption the final statement in the lemma follows. 
\end{proof}

\begin{proof}[Proof of \Cref{prop2}]
Denote voter $i$'s set-extension value by $j^*$. Let $A_i$ be voter $i$'s MBR ballot leading to election outcome $W$ which maximises utility $U_i(W)$. Denote the elements of $A_i$ as
$$A_i=\{c_1, \ldots, c_\ell\},$$
such that $c_1\spref_i c_2\spref_i\cdots \spref_ic_\ell$. Note that $|A_i|=\ell\le j^*$ by \Cref{lem2}.

The proof is inductive, and so we show that adding a candidate $c\spref_i c_\ell$ and $c\notin A_i$ to voter $i$'s ballot can not reduce the utility and hence remains a best-response.  First note that if $c\in W$, then $A_i\cup\{c\}$ does not change the election outcome, by the monotonically robust property, and so it is also a best response and we are done.

Now, let $c\notin W$ such that $c\spref_i c_\ell$ and consider 
$$A_i'=A_i\cup\{c\},$$
with outcome $W'$. Since $A_i$ is a best response it must be that $U_i(W')\le U_i(W)$. If equality holds, then we are done. For the purpose of a contradiction suppose that strict inequality holds i.e.
\begin{align}\label{starhash}
U_i(W')<U_i(W),
\end{align}
and so $W'\neq W$. Equation~(\ref{starhash}) combined with the monotonically robust property implies that
$$W'=(W\cup\{c\})\backslash \{c'\},$$
for some $c'\spref_i c\spref_i c_\ell$, so that $c$ replaces a more preferred candidate $c'$ in the winning committee. Note that $c_\ell \in W'$ since $c'\neq c_\ell$. 

Now consider $A_i^{-1}=A_i\backslash \{c_\ell\}$ with outcome $W^{-1}$. Since $A_i$ is a MBR and $|A_i^{-1}|<|A_i|$ it must be that $U_i(W^{-1})<U_i(W)$ and so $W^{-1}\neq W$. Again by the monotonically robust property it must be that
\begin{align}\label{MBRlink}
W^{-1}=(W\cup\{c_L\})\backslash \{c_\ell\},
\end{align}
for some $c_L$ such that $c_\ell\spref_i c_L$. That is, $c_\ell$ is replaced with a less preferred candidate $c_L$ in the winning committee. Note that $c_\ell\notin W^{-1}$. 

Now consider $\hat{A}_i=A_i^{-1}\cup\{c\}$ with outcome $\hat{W}$. Since $\hat{A}_i$ is simply a reinforcement of $c$ from $A_i^{-1}$ it must be that
\begin{align}\label{contr1}
\hat{W}&=\begin{cases}
W^{-1} &\text{or, }\\
W^{-1}\cup\{c\}\backslash \{e\} &\text{for some $e\in W^{-1}$,}
\end{cases}
\end{align}
by the monotonically robust property. That is, either no change occurs to the winning committee, or $c$ is added to the committee and some other previously elected candidate is removed. But note also that $\hat{A}_i=A_i'\backslash \{c_\ell\}$, and so $A_i'$ is a reinforcement of $c_\ell$ from $\hat{A}_i$, this implies that (again by the monotonically robust property)
\begin{align}\label{contr2}
W'&=\begin{cases}
\hat{W} &\text{or, }\\
\hat{W}\cup\{c_\ell\}\backslash \{f\} &\text{for some $f\in \hat{W}$.}
\end{cases}
\end{align}
But recall that $c_\ell \notin W^{-1}$ and $c\neq c_\ell$ and so by (\ref{contr1})
\begin{align}\label{contr3}
c_\ell \notin \hat{W}.
\end{align}
Now equation~(\ref{contr2}) implies that $W'=\hat{W}$. But $c_\ell\in W'$ and so $c_\ell\in \hat{W}$ which contradicts (\ref{contr3}).

This shows the required result when $A_i$ is a MBR, inductively we can repeat this process to generate a sincere and best response ballot. 

The only point of clarification required is in step (\ref{MBRlink}), however any updated $A_i$ which is a BR but no longer MBR will still satisfy this result. This follows from the fact that replacing $A_i$ with $A_i\cup \{c\}$ means that $A_i^{-1}=A_i\cup\{c\}\backslash \{c_\ell\}$ is simply a reinforcement of $c$ with respect to the original MBR ballot. We conclude that the only possible change is that $c\in W^{-1}$ - which means that $c_\ell$ will still never be in the outcome $W^{-1}$. Thus, the induction holds. 
\end{proof}


		\end{document}